\documentclass[runningheads]{llncs}

\usepackage[T1]{fontenc}
\usepackage{graphicx}
\usepackage{amsmath,amssymb}
\usepackage{url}
\usepackage{tikz}
\usetikzlibrary{arrows.meta,positioning,calc}
\usepackage{xcolor}
\usepackage{booktabs}
\usepackage{enumitem}
\usepackage[labelformat=simple]{subcaption}

\definecolor{copenblue}{RGB}{0,102,204}

\newcommand{\E}{\mathbb{E}}

\newcommand{\Z}{\mathcal{Z}}
\newcommand{\cX}{\mathcal{X}}
\newcommand{\cH}{\mathcal{H}}

\begin{document}

\title{Rate-Distortion Function for Encrypted Traffic Side-Channel Defense}

\titlerunning{Rate-Distortion Function for Encrypted Traffic Defense}

\author{Guangjie Liu\inst{1} \and
Guang Cheng\inst{2} \and
Weiwei Liu\inst{3} \and
Yutong Wang\inst{1}}

\authorrunning{G. Liu et al.}

\institute{
  School of Electronics and Information Engineering,
  Nanjing University of Information Science and Technology,
  Nanjing 210044, China\\
  \email{gjieliu@gmail.com}
\and
  School of Cyber Science and Engineering,
  Southeast University, Nanjing 211189, China\\
  \email{chengguang@seu.edu.cn}
\and
  School of Automation,
  Nanjing University of Science and Technology, Nanjing 210094, China \\
  \email{lwwnjust@njust.edu.cn}
}

\maketitle

\begin{abstract}
Parameter selection for encrypted traffic defense has long relied on empirical
tuning, yet the fundamental question---\emph{given a QoS cost budget $D$, how
low can the leakage rate go under sustained observation?}---lacks a provable,
computable baseline.  Taking the semantic label sequence $X^n$ as the source,
the defended feature sequence $Y^n$ as the observation, and Wasserstein-1
distance as the defense cost, we define the \emph{side-channel rate-distortion
function} $R^{\mathrm{sc}}(D)$ within the stationary memoryless defense class
$\Theta_{\mathrm{iid}}$ and provide its complete characterization.  We prove
that $R^{\mathrm{sc}}(D)$ is monotone decreasing, convex, and continuous, with
exact endpoints; the optimal defense has an exponential-tilting (Boltzmann)
structure governed by KKT conditions; and the curve constitutes the exact
Pareto frontier within $\Theta_{\mathrm{iid}}$.  For binary equal-prior tasks,
$D_{\max} = \tfrac{1}{2}W_1(P_0,P_1)$ via Kantorovich--Rubinstein duality.  On
real-world website-fingerprinting defenses, the framework locates Front
($\Delta_{\mathrm{gap}}{=}0.028$\,bits), WTF-PAD ($0.034$\,bits), and
TrafficSliver ($0.124$\,bits) above the theoretical curve, quantifying their
suboptimality gaps.

\keywords{Encrypted traffic side channel \and Rate-distortion theory \and
Optimal transport \and Wasserstein distance \and Website fingerprinting}
\end{abstract}

\section{Introduction}
\label{sec:intro}

Encrypted network traffic carries rich metadata---packet lengths, inter-packet
timings, and flow directions---that leaks application semantics even when
payload is hidden.  Website fingerprinting (WF) attacks exploit this channel
with alarming accuracy: closed-world recognition of Tor-browsed sites reaches
91--96\%~\cite{wang2014effective,sirinam2018deep}, application identification
exceeds 90\%~\cite{shen2021graphdapp}, and real-gateway accuracy remains above
93\% at very low base rates~\cite{mei2025high}.  Modern defenses (Tamaraw,
WTF-PAD~\cite{juarez2016wtfpad}, Front~\cite{gong2020zero},
TrafficSliver~\cite{delacadena2020trafficsliver},
NetShaper~\cite{sabzi2024netshaper}) trade bandwidth and latency for privacy,
but parameter selection is empirical and protection strength is measured by
adversarial accuracy.

\textbf{The missing baseline.}  A systematic re-evaluation of nine prominent
defenses~\cite{mathews2023sok} concluded that the absence of formal,
provable guarantees is the root cause of defenses being repeatedly broken.
Liu~et~al.~\cite{liu2026inevitability} gave an information-theoretic
diagnosis: side-channel leakage is \emph{inevitable}---for any
first-order distinguishable application pair, $I(X;Z)>0$ is a structural
consequence regardless of the defense.  Yet inevitability only answers whether
leakage is zero.  The central quantitative question remains open:
\begin{quote}
  \emph{Given QoS cost budget $D$, how low can the leakage rate
  go under sustained observation?}
\end{quote}
No existing tool---differential privacy~\cite{dwork2006distributed},
application-specific bandwidth
bounds~\cite{cai2014systematic},
or information-theoretic leakage
metrics~\cite{serjantov2002information,cherubin2017bayes}---provides a
computable answer to this general question.

\textbf{Our approach.}  This paper answers the question by casting it as
a \emph{dual of Shannon rate-distortion}: instead of minimizing description
rate subject to a distortion constraint, we minimize leakage rate subject to a
defense cost constraint.  The mathematical skeleton is identical; only the
roles of rate and distortion are swapped.  We use Wasserstein-1 ($W_1$)
distance as the defense cost measure because (i)~it respects the geometry of
the feature space, (ii)~the $W_1$ ball is convex, enabling convex optimization,
and (iii)~the Kantorovich--Rubinstein duality couples $W_1$ naturally to
$L$-Lipschitz statistics, providing a common language for attacker and
defender.

\textbf{Contributions.}  We define the side-channel rate-distortion function
\begin{equation}\label{eq:rsc}
  R^{\mathrm{sc}}(D) = \min_{\{Q_x\}:\,\sum_x p_X(x)W_1(Q_x,P_x)\le D}
  I(X;Y),
\end{equation}
and establish the following results.
\begin{enumerate}[label=(\arabic*),topsep=2pt,itemsep=1pt]
  \item \textbf{Properties} (Theorem~\ref{thm:properties}): $R^{\mathrm{sc}}(D)$
    is monotone decreasing, convex, and continuous with exact endpoints
    $R^{\mathrm{sc}}(0)=I(X;Z)$ and $R^{\mathrm{sc}}(D_{\max})=0$.
  \item \textbf{Optimal defense structure} (Theorem~\ref{thm:kkt}): The
    optimal defense has an exponential-tilting form
    $Q_x^*(y)\propto \bar{Q}^*(y)\exp(-\lambda^*\phi_x^*(y))$; the KKT
    multiplier $\lambda^*$ quantifies the marginal value of defense resources.
  \item \textbf{Pareto frontier} (Theorem~\ref{thm:pareto}): $(D, R^{\mathrm{sc}}(D))$
    is the exact achievable boundary within $\Theta_{\mathrm{iid}}$.
  \item \textbf{Critical cost formula} (Proposition~\ref{prop:dmax}):
    For binary equal-prior tasks,
    $D_{\max} = \tfrac{1}{2}W_1(P_0,P_1)$, bridging the zero-leakage
    threshold directly to the statistical separability of the two
    application distributions.
\end{enumerate}

\section{Related Work}
\label{sec:related}
\begin{sloppypar}

\textbf{Traffic analysis attacks.}
Website fingerprinting exploits packet metadata (lengths, timings, directions)
to identify encrypted communication.  Wang~et~al.~\cite{wang2014effective}
demonstrated 85\% TPR at 0.6\% FPR in a large open-world setting;
k-fingerprinting~\cite{hayes2016kfingerprinting} pushed this to 0.02\% FPR
across 100{,}000 pages.  Deep Fingerprinting~\cite{sirinam2018deep} achieves
98\% closed-world accuracy on Tor; ET-BERT~\cite{lin2022etbert} shows
cross-protocol distinguishability via pre-trained representations.
GraphDApp~\cite{shen2021graphdapp} exceeds 90\% accuracy in application
identification; real-gateway evaluations confirm accuracy above 93\% at very
low base rates~\cite{mei2025high}.  Siby~et~al.~\cite{siby2023quic} show
that even 10\%-sampled QUIC traffic yields recognition far above chance.
The consistent cross-scenario finding is that different applications produce
statistically distinguishable feature distributions $\{P_x\}$, making
$I(X;Z)>0$ a structural fact and the fundamental threat model well-founded.

\textbf{Traffic defenses and the formal-guarantee gap.}
Early defenses established baseline tradeoffs: BuFLO~\cite{dyer2012peek}
eliminates distinguishing features at over 100\% bandwidth overhead;
Tamaraw~\cite{cai2014systematic} derives the first scenario-specific bandwidth
lower bound for $\varepsilon$-secure defenses; Walkie-Talkie~\cite{wang2017walkie} reduces overhead
to ${\approx}31\%$ via a half-duplex schedule.  WTF-PAD~\cite{juarez2016wtfpad}
applies adaptive inter-burst padding at under 60\% overhead, but was
subsequently broken by Deep Fingerprinting~\cite{sirinam2018deep}.
Front~\cite{gong2020zero} achieves zero latency via Rayleigh-distributed
dummy injection; TrafficSliver~\cite{delacadena2020trafficsliver} splits flows
across guard nodes to reduce the per-node observable;
NetShaper~\cite{sabzi2024netshaper} frames side-channel mitigation as
$(\varepsilon,\delta)$-differential privacy and provides a quantified
privacy-overhead curve---but DP parameters accumulate linearly under
repeated composition, and the framework does not address the rate-distortion
question of how low leakage can go for a given budget $D$.
Palette~\cite{shen2024palette} clusters traffic patterns and normalizes to
a common template, reducing attack accuracy by ${\approx}74\%$ on average,
yet its guarantees depend on cluster radius rather than a universal bound.
Mathews~et~al.~\cite{mathews2023sok} re-evaluated nine defenses and found
that most fail under adaptive attacks; they explicitly attribute this to
the lack of formal, provable guarantees---the precise gap our work addresses.

\textbf{Information-theoretic leakage metrics.}
Serjantov and Danezis~\cite{serjantov2002information} introduced mutual
information as an anonymity metric; Chatzikokolakis~et~al.\ extended this to
a unified $g$-leakage framework~\cite{chatzikokolakis2007anonymity}.
Cherubin~\cite{cherubin2017bayes} connected Bayes error to mutual information
via Fano-type inequalities, placing attack accuracy and information leakage
in a common framework.  Liu~et~al.~\cite{liu2026inevitability} proved the
\emph{inevitability} of side-channel leakage: for any first-order
distinguishable application pair, $I(X;Y)>0$ holds for every defense---a
structural result that motivates asking not whether leakage is zero, but
how small it can be made.

\textbf{Rate-distortion theory and optimal transport.}
Shannon~\cite{shannon1959coding} and Berger~\cite{berger1971rate} established
rate-distortion theory; the classical problem minimizes description rate
subject to a distortion bound.  Yamamoto~\cite{yamamoto1997rate} and
Merhav--Shamai~\cite{merhav2007information} studied rate-distortion under
secrecy constraints (state masking), closest in spirit to our work.
The Blahut--Arimoto algorithm~\cite{blahut1972computation} provides a
convergent iterative solver for classical rate-distortion; our SLSQP-based
approach extends this to the composite $W_1$ distortion.
Villani~\cite{villani2009optimal} established the Kantorovich--Rubinstein
duality underlying our cost definition; Wasserstein distance has been applied
to differential privacy~\cite{wassersteindp2024} but not to the
rate-distortion characterization of traffic defense.  The present work
bridges Shannon rate-distortion~\cite{shannon1959coding,berger1971rate}
with optimal transport~\cite{villani2009optimal} and instantiates the
combined framework on encrypted traffic for the first time.

\textbf{Positioning.}  Three specific gaps motivate this work.
(1)~\emph{No cost-leakage tradeoff curve exists}: attack literature confirms
$I(X;Z)>0$; defense literature accumulates empirical tradeoff curves; but
neither provides the theoretical lower bound on leakage as a function of $D$.
(2)~\emph{Optimal defense structure is unknown}: existing defenses are
heuristic and do not align perturbations with any optimality criterion.
(3)~\emph{No quantitative bridge to the theoretical limit}: Mathews
et~al.~\cite{mathews2023sok} identify the gap but provide no metric for
measuring distance from the optimum.  We fill all three gaps.

\end{sloppypar}

\section{Problem Formulation}
\label{sec:model}

\subsection{Model}

\begin{definition}[Semantic source]
$\{X_i\}_{i\ge 1}$ is i.i.d.\ with $X_i\in\cX=\{1,\ldots,K\}$ and prior
$p_X$, where $\min_x p_X(x)>0$.
\end{definition}

\begin{definition}[Feature space and original distributions]
The packet feature space $\Z$ is finite with metric $d:\Z\times\Z\to[0,1]$.
The original conditional distribution $P_x = P_{Z|X=x}$ is the
feature distribution of application $x$ in the absence of any defense.
\end{definition}

A \emph{defense strategy} $\theta = \{Q_x\}_{x\in\cX}$ specifies a target
defended distribution: $Y|X=x \sim Q_x$.  We study the \emph{stationary
memoryless defense class}
\[
  \Theta_{\mathrm{iid}} \triangleq
  \bigl\{\theta = \{Q_x\} : Q_x \in \mathcal{P}(\Z),\
  Y_i|X_i=x \overset{\mathrm{iid}}{\sim} Q_x\ \text{across}\ i \bigr\}.
\]
This class covers deterministic traffic shaping and random padding; defenses
with flow-level shared randomness or cross-packet memory are discussed in
Sect.~\ref{sec:discussion}.

\subsection{Defense Cost and Leakage Rate}

\begin{definition}[Defense cost]
\label{def:cost}
\[
  D(\theta) \triangleq \sum_{x\in\cX} p_X(x)\,W_1(Q_x, P_x),
\]
where $W_1(Q_x,P_x) = \min_{\gamma\in\Pi(Q_x,P_x)}\sum_{y,z}d(y,z)\gamma(y,z)$
is the Wasserstein-1 distance.  $D(\theta)$ is the average cost of
``transporting'' each application's feature distribution from $P_x$ to $Q_x$.
\end{definition}

We choose $W_1$ over KL divergence for three reasons:
(i)~$W_1$ respects the geometry of $\Z$; (ii)~the $W_1$ ball $\{Q:W_1(Q,P_x)\le r\}$
is convex, so the feasible set of~\eqref{eq:rsc} is convex; and
(iii)~the Kantorovich--Rubinstein duality $W_1(Q,P)=\sup_{h\in\cH_1}|\E_Q h - \E_P h|$
couples the defense cost directly to $L$-Lipschitz statistics used by the
attacker.

\begin{definition}[Leakage rate]
$R(\theta) \triangleq \lim_{n\to\infty}\tfrac{1}{n}I(X^n;Y^n)$.
\end{definition}

Using a single-sample $X$ as the semantic variable leads to the trivial
bound $\frac{1}{n}I(X;Y^n)\le \frac{\log K}{n}\to 0$; using the
source sequence $X^n$ yields a meaningful per-step limit.

\begin{lemma}[Single-letterization]
\label{lem:single-letter}
For $\theta\in\Theta_{\mathrm{iid}}$, the limit exists and equals
$R(\theta) = I(X;Y)$, where $(X,Y)\sim p_X(x)Q_x(y)$.
\end{lemma}
\begin{proof}
  Under $\Theta_{\mathrm{iid}}$, $(X_i,Y_i)$ are i.i.d., so
  $I(X^n;Y^n) = \sum_{i=1}^n I(X_i;Y_i) = n\,I(X;Y)$.  Dividing by $n$
  gives the result.
\end{proof}

\subsection{Side-Channel Rate-Distortion Function}

\begin{definition}
\label{def:rd}
For $D\ge 0$:
\begin{equation}
  R^{\mathrm{sc}}(D) \triangleq
  \min_{\{Q_x\}:\,\sum_x p_X(x)W_1(Q_x,P_x)\le D} I(X;Y).
\end{equation}
\end{definition}

Since $\Z$ is finite, this is a finite-dimensional convex program: $I(X;Y)$
is convex in $\{Q_x\}$ (for fixed $p_X$), and the constraint function is
convex.  The minimum is attained (compact feasible set, continuous objective).

\section{Main Results}
\label{sec:results}

\subsection{Properties of $R^{\mathrm{sc}}(D)$}

\begin{theorem}[Basic properties]
\label{thm:properties}
$R^{\mathrm{sc}}(D)$ satisfies:
\begin{enumerate}[label=(\roman*),topsep=2pt,itemsep=1pt]
  \item \textup{(Monotone)} $D_1<D_2 \Rightarrow R^{\mathrm{sc}}(D_1)\ge R^{\mathrm{sc}}(D_2)$.
  \item \textup{(Convex)} $R^{\mathrm{sc}}(\lambda D_1+(1-\lambda)D_2)\le
        \lambda R^{\mathrm{sc}}(D_1)+(1-\lambda)R^{\mathrm{sc}}(D_2)$.
  \item \textup{(Continuous)} $R^{\mathrm{sc}}$ is continuous on $[0,D_{\max}]$
        and zero for $D>D_{\max}$.
  \item \textup{(Endpoints)} $R^{\mathrm{sc}}(0) = I(X;Z)$ and
        $R^{\mathrm{sc}}(D)=0 \Leftrightarrow D\ge D_{\max}$, where
        \[
          D_{\max} \triangleq \inf_{Q\in\mathcal{P}(\Z)}
          \sum_x p_X(x)\,W_1(Q,P_x).
        \]
\end{enumerate}
\end{theorem}

\begin{proof}
\textit{(i) Monotonicity.}
If $D_1<D_2$, the feasible set $\mathcal{F}(D_1)\subseteq\mathcal{F}(D_2)$,
so minimizing over a larger set cannot increase the value.

\textit{(ii) Convexity.}
Let $\{Q_x^{(j)}\}$ be optimal at level $D_j$ for $j=1,2$, and set
$Q_x^{(\lambda)}=\lambda Q_x^{(1)}+(1{-}\lambda)Q_x^{(2)}$.
By convexity of $W_1(\cdot,P_x)$ (linear program value in first argument):
$\sum_x p_x W_1(Q_x^{(\lambda)},P_x)\le\lambda D_1+(1{-}\lambda)D_2$.
By convexity of $I(X;Y)$ in the channel $\{Q_x\}$ at fixed $p_X$~\cite{cover2006elements}:
$I_\lambda(X;Y)\le\lambda R^{\mathrm{sc}}(D_1)+(1{-}\lambda)R^{\mathrm{sc}}(D_2)$.
Hence $R^{\mathrm{sc}}(\lambda D_1+(1{-}\lambda)D_2)\le\lambda R^{\mathrm{sc}}(D_1)+(1{-}\lambda)R^{\mathrm{sc}}(D_2)$.

\textit{(iii) Continuity.}
A finite-valued convex function on an open interval is continuous there, so
continuity holds on $(0,D_{\max})$.  At $D=0$: monotonicity gives
$\lim_{D\downarrow 0}R^{\mathrm{sc}}(D)\le R^{\mathrm{sc}}(0)$.  For the reverse,
any optimal $\{Q_x^D\}$ satisfies $W_1(Q_x^D,P_x)\le D/(\min_x p_x)\to 0$,
and since $\Z$ is finite, $W_1\to 0$ implies $Q_x^D\to P_x$, so
$I_D(X;Y)\to I(X;Z)=R^{\mathrm{sc}}(0)$.  At $D_{\max}$: if the left limit
$L^->\!0$, convexity applied to $D_1<D_{\max}$ and $D_2=D_{\max}$ gives
$R^{\mathrm{sc}}(\tfrac{D_1+D_{\max}}{2})\le\tfrac{1}{2}R^{\mathrm{sc}}(D_1)$;
letting $D_1\uparrow D_{\max}$ yields $L^-\le\frac{1}{2}L^-$, a contradiction.

\textit{(iv) Left endpoint.}
$D=0$ forces $W_1(Q_x,P_x)=0$ for every $x$.  On finite $\Z$, this means
$Q_x=P_x$, so $Y|X=x\sim P_x$ and $I(X;Y)=I(X;Z)$.

\textit{(v) Zero-leakage condition.}
($\Leftarrow$) $\mathcal{P}(\Z)$ is compact and $Q\mapsto\sum_x p_x W_1(Q,P_x)$
is continuous, so $D_{\max}$ is attained at some $Q^\dagger$.
Setting $Q_x\equiv Q^\dagger$ for all $x$ gives cost $D_{\max}\le D$
and $I(X;Y)=0$ (since $Y\perp X$).
($\Rightarrow$) $I(X;Y)=0$ implies $D_{\mathrm{KL}}(Q_x^*\|p_Y)=0$
for all $x$ with $p_x>0$, so $Q_x^*=p_Y$ for all such $x$.
Thus $D\ge\sum_x p_x W_1(p_Y,P_x)\ge D_{\max}$.
\end{proof}

Properties (i)--(ii) encode diminishing returns: early budget buys
substantial leakage reduction, while zero leakage requires disproportionate cost.

\begin{corollary}[Marginal value of defense]
\label{cor:subgrad}
For $D\in(0,D_{\max})$, let $\lambda^*(D)$ be the optimal dual multiplier.
Then $-\lambda^*(D)\in\partial R^{\mathrm{sc}}(D)$: $\lambda^*$ is the
marginal leakage reduction per unit of defense cost.  Wherever
$R^{\mathrm{sc}}$ is differentiable,
$\tfrac{d}{dD}R^{\mathrm{sc}}(D)=-\lambda^*(D)<0$.
\end{corollary}

\subsection{Optimal Defense: Exponential-Tilting Structure}

\begin{theorem}[KKT conditions and exponential tilting]
\label{thm:kkt}
Let $D\in(0,D_{\max})$ and let $\{Q_x^*\}$ be optimal for~\eqref{eq:rsc}
with $\bar{Q}^* = \sum_x p_X(x)Q_x^*$.  Then there exists a dual multiplier
$\lambda^*>0$ and, for each $x$, a subgradient
$\phi_x^*\in\partial W_1(\cdot,P_x)|_{Q_x^*}$ (the Kantorovich potential)
such that on the support $\mathcal{S}_x\triangleq\{y:Q_x^*(y)>0\}$:
\begin{equation}
  \log Q_x^*(y) - \log\bar{Q}^*(y) + \lambda^*\phi_x^*(y) + c_x = 0,
\end{equation}
equivalently,
\begin{equation}\label{eq:exp-tilt}
  Q_x^*(y) \propto \bar{Q}^*(y)\,\exp\bigl(-\lambda^*\phi_x^*(y)\bigr),
  \quad y\in\mathcal{S}_x.
\end{equation}
The constraint is active: $\sum_x p_X(x)W_1(Q_x^*,P_x)=D$.
\end{theorem}

\begin{proof}
\textit{Existence.}
$\prod_x\mathcal{P}(\Z)$ is compact in the finite-dimensional topology,
$I(X;Y)$ and $W_1$ are continuous, so the minimum is attained
(Weierstrass theorem).

\textit{Strong duality.}
Taking $Q_x=P_x$ gives $D(\theta)=0<D$, so Slater's condition holds.
For a finite-dimensional convex program with a Slater point, strong duality
holds and the KKT conditions are necessary and sufficient for optimality.

\textit{Lagrangian and stationarity.}
Introduce multiplier $\lambda\ge 0$ (inequality constraint),
$\nu_x\in\mathbb{R}$ (normalization $\sum_y Q_x(y)=1$),
and $\eta_x(y)\ge 0$ (non-negativity $Q_x(y)\ge 0$).
Write $I(X;Y)=\sum_x p_x\sum_y Q_x(y)\log Q_x(y)-\sum_y\bar{Q}(y)\log\bar{Q}(y)$
where $\bar{Q}(y)=\sum_x p_x Q_x(y)$.  For $y\in\mathcal{S}_x$ (where
$Q_x^*(y)>0$, so $\eta_x(y)=0$ by complementary slackness),
differentiating with respect to $Q_x(y)$ and using
$\partial\bar{Q}(y)/\partial Q_x(y)=p_x$ gives:
\[
  \frac{\partial I}{\partial Q_x(y)} = p_x\bigl(\log Q_x(y)-\log\bar{Q}(y)\bigr).
\]
For the $W_1$ term, since $W_1(\cdot,P_x)$ is a convex piecewise-linear
function (LP value), we select a subgradient $\phi_x^*$ (the optimal
Kantorovich dual potential).  The stationarity condition
$\nabla_{Q_x(y)}\mathcal{L}=0$ becomes:
\[
  p_x\bigl(\log Q_x^*(y)-\log\bar{Q}^*(y)+\lambda^*\phi_x^*(y)\bigr)+\nu_x=0.
\]
Setting $c_x=\nu_x/p_x$ and exponentiating yields~\eqref{eq:exp-tilt}.

\textit{Constraint activation.}
Suppose $D<D_{\max}$ but the constraint is inactive.  Then $\lambda^*=0$
by complementary slackness, and the stationarity condition reduces to
$\log Q_x^*(y)=\log\bar{Q}^*(y)-c_x$, i.e.\ $Q_x^*\propto\bar{Q}^*$,
which after normalization gives $Q_x^*=\bar{Q}^*$ for all $x$.
This implies $I(X;Y)=0$, so the cost of this solution satisfies
$\sum_x p_x W_1(\bar{Q}^*,P_x)\ge D_{\max}>D$,
contradicting feasibility.  Hence $\lambda^*>0$ and the constraint is active.
\end{proof}

Equation~\eqref{eq:exp-tilt} is the traffic-defense analog of the
Boltzmann distribution in rate-distortion
theory~\cite{cover2006elements,gallager1968information}: $\phi_x^*$ acts
as the per-class cost function and $\lambda^*$ is the marginal value
of defense resources.  A Blahut--Arimoto-style~\cite{blahut1972computation}
sweep over $\lambda$ traces the full $R^{\mathrm{sc}}(D)$ curve.

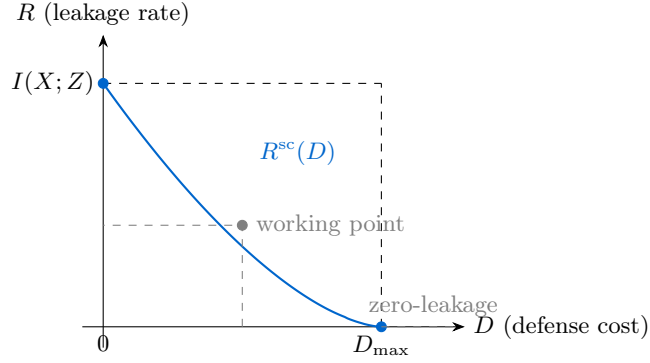
\begin{figure}[tb]
\centering
\begin{tikzpicture}[>=Stealth,scale=0.92]
  \draw[->] (-0.3,0) -- (5.2,0) node[right] {$D$ (defense cost)};
  \draw[->] (0,-0.3) -- (0,4.2) node[above] {$R$ (leakage rate)};
  \draw[dashed] (0,3.5) -- (4.0,3.5);
  \draw[dashed] (4.0,0) -- (4.0,3.5);
  \node[left] at (0,3.5) {$I(X;Z)$};
  \node[below] at (4.0,0) {$D_{\max}$};
  \node[below] at (0,0) {$0$};
  \draw[thick,color=copenblue,domain=0:4.0,samples=100]
    plot (\x, {3.5*(1-\x/4.0)^1.6});
  \filldraw[copenblue] (0,3.5) circle (2pt);
  \filldraw[copenblue] (4.0,0) circle (2pt);
  \draw[dashed,gray] (2.0,0) -- (2.0,1.46) -- (0,1.46);
  \filldraw[gray] (2.0,1.46) circle (2pt)
    node[right,xshift=2pt] {\small working point};
  \node[color=copenblue] at (2.8,2.5) {$R^{\mathrm{sc}}(D)$};
  \draw[gray,dashed,thin] (4.0,0) -- (5.1,0);
  \node[gray,font=\small] at (4.75,0.3) {zero-leakage};
\end{tikzpicture}
\caption{Side-channel rate-distortion curve $R^{\mathrm{sc}}(D)$: the
exact Pareto frontier within $\Theta_{\mathrm{iid}}$.  Any defense in
the class operates at or above the curve.}
\label{fig:rd-curve}
\end{figure}

\subsection{Pareto Frontier and Critical Cost}

\begin{theorem}[Exact Pareto frontier within $\Theta_{\mathrm{iid}}$]
\label{thm:pareto}
\begin{enumerate}[label=(\roman*),topsep=2pt,itemsep=1pt]
  \item \textup{(Achievability)} For every $D\in[0,D_{\max}]$, there exists
    $\theta^*\in\Theta_{\mathrm{iid}}$ with $D(\theta^*)\le D$ and
    $R(\theta^*)=R^{\mathrm{sc}}(D)$.  For $D\in(0,D_{\max})$ the constraint
    is active.
  \item \textup{(Unimprovability)} For any $\theta\in\Theta_{\mathrm{iid}}$
    with $D(\theta)\le D$, we have $R(\theta)\ge R^{\mathrm{sc}}(D)$.
\end{enumerate}
\end{theorem}

\begin{proof}
\textit{(i) Achievability.}
The feasible set is compact and $I(X;Y)$ is continuous, so the infimum
in Definition~\ref{def:rd} is attained for every $D\ge 0$.
Constraint activation for $D\in(0,D_{\max})$ follows from Theorem~\ref{thm:kkt}.

\textit{(ii) Unimprovability.}
$R^{\mathrm{sc}}(D)$ is defined as the infimum of $R(\theta)$ over all
$\theta\in\Theta_{\mathrm{iid}}$ with $D(\theta)\le D$.
By definition, $R(\theta)\ge R^{\mathrm{sc}}(D)$ for every such $\theta$.
\end{proof}

These two properties jointly say that $R^{\mathrm{sc}}(D)$ is a tight lower
bound: no i.i.d.\ defense can do better, and the bound is achieved by
a concrete distribution family $\{Q_x^*\}$.

\begin{proposition}[Critical cost for binary equal-prior]
\label{prop:dmax}
Let $K=2$ and $p_X(0)=p_X(1)=\tfrac{1}{2}$.  Then $D_{\max}\ge\tfrac{1}{2}W_1(P_0,P_1)$.
If a $W_1$-geodesic midpoint $Q_{1/2}$ exists (i.e.,\
$W_1(P_0,Q_{1/2})=W_1(Q_{1/2},P_1)=\tfrac{1}{2}W_1(P_0,P_1)$), then
\begin{equation}\label{eq:dmax}
  D_{\max} = \tfrac{1}{2}W_1(P_0,P_1) = \frac{\Delta_{\mathrm{app}}^{(L)}}{2L},
\end{equation}
where $\Delta_{\mathrm{app}}^{(L)}=\sup_{h\in\cH_L}|\E_{P_0}h-\E_{P_1}h|$
and the equality follows from Kantorovich--Rubinstein duality.
\end{proposition}

\begin{proof}
\textit{Lower bound.}
By the $W_1$ triangle inequality, for any $Q$:
\[
  \tfrac{1}{2}\bigl(W_1(Q,P_0)+W_1(Q,P_1)\bigr)
  \ge \tfrac{1}{2}W_1(P_0,P_1).
\]
Since $\mathcal{P}(\Z)$ is compact and $Q\mapsto\sum_x p_x W_1(Q,P_x)$
is continuous, the infimum over $Q$ is attained and
$D_{\max}\ge\tfrac{1}{2}W_1(P_0,P_1)$.

\textit{Upper bound (geodesic midpoint).}
If $Q_{1/2}$ exists, choose $Q_0=Q_1\equiv Q_{1/2}$ as the common defended
distribution for both classes.  Then
\[
  D_{\max}\le\tfrac{1}{2}\bigl(W_1(P_0,Q_{1/2})+W_1(P_1,Q_{1/2})\bigr)
  =\tfrac{1}{2}W_1(P_0,P_1).
\]
Combining both bounds gives $D_{\max}=\tfrac{1}{2}W_1(P_0,P_1)$.

\textit{Kantorovich--Rubinstein link.}
On finite $\Z$ with metric $d$, K-R duality~\cite{villani2009optimal} states
$W_1(P_0,P_1)=\sup_{h\in\cH_1}|\E_{P_0}h-\E_{P_1}h|$.  For general
Lipschitz constant $L$, rescaling gives
$\Delta_{\mathrm{app}}^{(L)}=L\,W_1(P_0,P_1)$, hence
$D_{\max}=\Delta_{\mathrm{app}}^{(L)}/(2L)$.
\end{proof}

Equation~\eqref{eq:dmax} ties $D_{\max}$ to a measurable quantity: the more
distinguishable the two applications, the higher the cost threshold for
zero leakage and the wider the useful budget range $[0,D_{\max})$.

\section{Empirical Evaluation}
\label{sec:exp}

\subsection{Experimental Setup}

\textbf{Dataset.}  We use the public website fingerprinting dataset of
Deng~et~al.~\cite{deng2024earlystage}, which is built on 95~monitored
websites collected by Sirinam~et~al.~\cite{sirinam2018deep} on Tor
(1,000~traces per site).  Four conditions are included: the undefended baseline (CW)
and three defenses: Front~\cite{gong2020zero}, WTF-PAD~\cite{juarez2016wtfpad},
and TrafficSliver~\cite{delacadena2020trafficsliver}.  \emph{Front} injects
dummy packets drawn from a Rayleigh distribution at zero latency cost.
\emph{WTF-PAD} uses an adaptive two-state automaton to inject dummy packets in
burst gaps.  \emph{TrafficSliver} splits Tor traffic across multiple guard
nodes (BWR strategy), so an adversary observing a single node sees only a
fraction of packets.

\textbf{Feature extraction.}  Raw traces are stored as signed-timestamp
sequences.  We extract inter-packet delays by differencing consecutive
non-zero timestamps (direction-agnostic) and quantize to $L=50$ uniform
histogram bins, obtaining empirical distributions $\widehat{P}_x$ and
defended distributions $\widehat{Q}_x^m$ per defense $m$.

\textbf{Class-pair selection.}  On the undefended CW data, we rank all
$\binom{95}{2}$ pairs by $W_1(\widehat{P}_0,\widehat{P}_1)$ and select the
top five: pairs (21,90), (4,21), (10,90), (83,90), (2,90), with
$W_1\in[0.039,0.042]$.  For each pair we assume equal prior
$p_X(0)=p_X(1)=\tfrac{1}{2}$ and compute the theoretical curve
$R^{\mathrm{sc}}(D)$ via Lagrange-multiplier sweep (SLSQP,
$N_\lambda=60$ log-uniform points).

\textbf{Defense evaluation.}  For each (class pair, defense) combination
the working-point coordinates are
\[
  \widehat{D}^m = \tfrac{1}{2}\bigl[W_1(\widehat{Q}_0^m,\widehat{P}_0)
    + W_1(\widehat{Q}_1^m,\widehat{P}_1)\bigr],\qquad
  \widehat{R}^m = I(X;\widehat{Y}^m).
\]
The \emph{suboptimality gap} is $\Delta_{\mathrm{gap}}^m =
\widehat{R}^m - R^{\mathrm{sc}}(\widehat{D}^m)\ge 0$, and the
\emph{cost utilization ratio} is $\widehat{D}^m/D_{\max}$.
Confidence intervals are obtained via 200-round flow-level bootstrap
with hierarchical averaging across the five class pairs.

\subsection{Results}

Table~\ref{tab:results} summarizes the average suboptimality gap for each
defense, and Fig.~\ref{fig:exp2} shows the two-metric comparison across class
pairs.  All three defenses sit strictly above the theoretical curve
($\Delta_{\mathrm{gap}}>0$; 95\,\% confidence intervals do not include zero).
Pairwise bootstrap tests confirm that the gaps are mutually distinguishable.

\begin{table}[tb]
\centering
\caption{Suboptimality gaps of three defenses (five class pairs, 95\,\%
bootstrap CI).}
\label{tab:results}
\begin{tabular}{lccc}
\toprule
Defense & $\overline{\Delta}_{\mathrm{gap}}$ (bits) & 95\,\% CI (bits)
        & $\widehat{D}/D_{\max}$ \\
\midrule
Front         & 0.028 & [0.026,\;0.030] & 89.2\,\% \\
WTF-PAD       & 0.034 & [0.033,\;0.036] & 83.6\,\% \\
TrafficSliver & 0.124 & [0.117,\;0.130] & 46.4\,\% \\
\bottomrule
\end{tabular}
\end{table}

\begin{figure}[tb]
\centering
\includegraphics[width=0.82\linewidth]{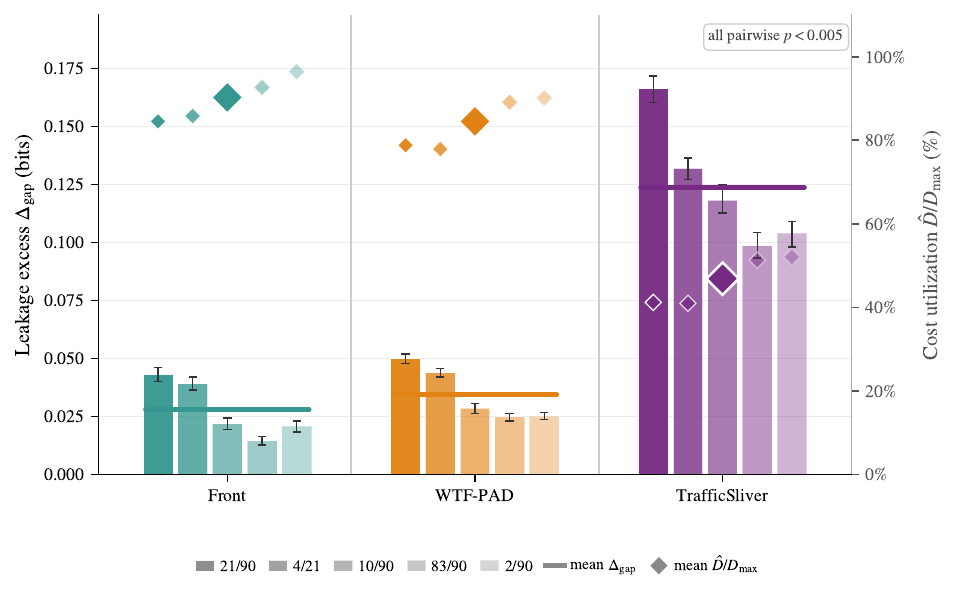}
\caption{Cost utilization vs.\ suboptimality gap for three defenses across
five class pairs (error bars: 95\,\% bootstrap CI).
Front and WTF-PAD consume $>$83\,\% of the available budget yet remain
suboptimal; TrafficSliver uses only 46\,\% of the budget and incurs the
largest gap, revealing inefficiency of the splitting strategy under this
metric.}
\label{fig:exp2}
\end{figure}

\textbf{Discussion.}  Front and WTF-PAD both approach the budget ceiling
($>$83\,\% utilization) but still leave measurable gaps.  This indicates that
the \emph{direction} of feature-space perturbation, not merely its magnitude,
matters---consistent with the exponential-tilting structure in
Theorem~\ref{thm:kkt}: optimal defenses must align perturbations with the
Kantorovich potential $\phi_x^*$, whereas Front and WTF-PAD use fixed
injection policies independent of the per-class cost gradient.
TrafficSliver's large gap ($0.124$ bits) at low cost utilization ($46\,\%$)
suggests that packet splitting reduces the \emph{observed} packet count but
does not reshape the timing distribution toward the optimal transport target.

\section{Discussion and Conclusion}
\label{sec:discussion}

\textbf{Scope of $R^{\mathrm{sc}}(D)$.}  The Pareto optimality result is exact
within $\Theta_{\mathrm{iid}}$.  For the broader class of \emph{memory-based}
defenses (e.g., those employing flow-level shared randomness $\kappa$),
a Jensen argument shows $\tfrac{1}{n}I(X^n;Y^n)\to I(X;Y|\kappa)
\ge R^{\mathrm{sc}}(D)$ asymptotically under per-packet marginal cost,
so $R^{\mathrm{sc}}(D)$ remains a valid lower bound.  Whether strictly
tighter bounds exist for adaptive causal defenses---requiring a multi-letter
rate-distortion formulation---is an open problem.

\textbf{Engineering gap.}  $D_{\max}$ is a mathematical threshold; reaching it
requires transporting all application distributions to a common target
$Q^\dagger$.  TCP acknowledgment constraints, buffer limits, and uplink/downlink
asymmetry all prevent this in practice, so real deployments operate in
$[0,D_{\max})$.  Quantifying the additional implementation penalty relative to
the distribution-layer baseline is a natural next step.

\textbf{Conclusion.}  We have established $R^{\mathrm{sc}}(D)$, the first
provable, computable baseline for the leakage-cost tradeoff in encrypted
traffic defense.  The curve is the exact Pareto frontier within
$\Theta_{\mathrm{iid}}$; the optimal defense follows an
exponential-tilting law; and $D_{\max}$ is tied to the Wasserstein distance
between application distributions.  Applied to three real-world defenses,
the framework quantifies suboptimality gaps---a principled benchmark against
which any future defense can be measured.

\begin{credits}
\subsubsection{\ackname}
This work was supported by the National Natural Science Foundation of China
Joint Fund Integration Project (No.~U2436601).

\subsubsection{\discintname}
The authors have no competing interests to declare.
\end{credits}

\bibliographystyle{splncs04}
\bibliography{reference}

\end{document}